\documentclass[12pt,a4paper]{article}
\usepackage{amsmath,amsthm,amsfonts,amssymb,bbm}
\usepackage{graphicx,psfrag,subfigure}
\usepackage{cite}
\usepackage{hyperref}
\usepackage{color}

\newcommand{\Id}{\mathbbm{1}}

\newcommand{\Pb}{\mathbbm{P}}

\newcommand{\e}{\varepsilon}
\newcommand{\I}{{\rm i}}
\newcommand{\C}{\mathbb{C}}
\newcommand{\R}{\mathbb{R}}
\newcommand{\N}{\mathbb{N}}
\newcommand{\Z}{\mathbb{Z}}
\newcommand{\dx}{\mathrm{d}}
\newcommand{\X}{\mathfrak X}

\renewcommand{\Re}{\mathrm{Re}}
\renewcommand{\Im}{\mathrm{Im}}

\DeclareMathOperator*{\diag}{diag}
\DeclareMathOperator*{\virt}{virt}
\DeclareMathOperator*{\Tr}{Tr}

\newtheorem{prop}{Proposition}[section]
\newtheorem{thm}[prop]{Theorem}
\newtheorem{lem}[prop]{Lemma}

\newtheorem{cla}[prop]{Claim}

\newtheorem{rem}[prop]{Remark}
\newenvironment{remark}{\begin{rem}\normalfont}{\end{rem}}

\title{On the partial connection between random matrices and interacting particle systems}

\author{Patrik L. Ferrari\thanks{Institute for Applied Mathematics, University of Bonn, Endenicher Allee 60,\newline 53115 Bonn, Germany; E-mail:~\texttt{ferrari@uni-bonn.de}},
Ren\'e Frings\thanks{Institute for Applied Mathematics, University of Bonn, Endenicher Allee 60,\newline 53115 Bonn, Germany; E-mail:~\texttt{frings@uni-bonn.de}}}

\date{10. September 2010}

\begin{document}
\sloppy
\maketitle

\begin{abstract}
In the last decade there has been increasing interest in the fields of random matrices, interacting particle systems, stochastic growth models, and the connections between these areas. For instance, several objects appearing in the limit of large matrices arise also in the long time limit for interacting particles and growth models. Examples of these are the famous Tracy-Widom distribution functions and the Airy$_2$ process.

The link is however sometimes fragile. For example, the connection between the eigenvalues in the Gaussian Orthogonal Ensembles (GOE) and growth on a flat substrate is restricted to one-point distribution, and the connection breaks down if we consider the joint distributions.

In this paper we first discuss known relations between random matrices and the asymmetric exclusion process (and a $2+1$-dimensional extension). Then, we show that the correlation functions of the eigenvalues of the matrix minors for $\beta=2$ Dyson's Brownian motion have, when restricted to increasing times and decreasing matrix dimensions, the same correlation kernel as in the $2+1$-dimensional interacting particle system under diffusion scaling limit. Finally, we analyze the analogous question for a diffusion on (complex) sample covariance matrices.
\end{abstract}

\newpage

\section{Introduction}
In the seminal paper~\cite{BDJ99} Baik, Deift, and Johansson prove that the longest increasing subsequence of a random permutation has fluctuations governed by the (GUE) Tracy-Widom distribution $F_2$. This distribution was discovered by Tracy and Widom as the one describing the fluctuations of the largest eigenvalue of random matrices from the Gaussian Unitary Ensemble (GUE)~\cite{TW94}.

Soon after, Johansson~\cite{Jo00b} showed that the same limiting distribution occurs in a stochastic growth model, which is equivalent to (a discrete time version of) the totally asymmetric simple exclusion process (TASEP) and belongs to Kardar-Parisi-Zhang universality class~\cite{KPZ86} of interacting particle systems~\cite{Li99}. This was the beginning of a lot of activities in this field located at the intersection between random matrices, stochastic growth models and interacting particle systems. For a recent review and a guide to literature on the subject, we refer to~\cite{FS10} or, for a review around TASEP, see~\cite{Fer07}. Our results are Theorem~\ref{ThmMainGUE} and Theorem~\ref{ThmMainWishart}, but before going in further explanations, we briefly define the models under consideration. We restrict the discussion to the TASEP and Gaussian ensembles of random matrices.

\subsubsection*{Gaussian ensembles of random matrices and Dyson's Brownian motion (DBM)}
(a) \emph{Hermitian matrices.}
The GUE ensemble of random matrices is defined\footnote{There are two other very common standard normalization in random matrix literature, see Table~1 in~\cite{Fer10}.} as the probability measure on $N\times N$ Hermitian matrices $H$ given by
\begin{equation}\label{eq1}
\frac{1}{Z_N} \exp\left(-\frac{\beta}{4N} \Tr(H^2)\right) \dx H,\quad \textrm{with}\quad \beta=2,
\end{equation}
where $\dx H=\prod_{i=1}^N \dx H_{i,i}\prod_{1\leq i<j\leq N} \dx\Re(H_{i,j}) \dx\Im(H_{i,j})$ and $Z_N$ is the normalization constant\footnote{Another way to describe (\ref{eq1}) is to take the upper-triangular entries to be independent and normal distributed: $H_{i,i}\sim {\cal N}(0,N)$ for $i=1,\dots,N$, while $\Re(H_{i,j})\sim {\cal N}(0,N/2)$ and $\Im(H_{i,j})\sim {\cal N}(0,N/2)$ for $1\leq i < j \leq N$.}. Notice that the measure (\ref{eq1}) is unitary invariant.

Dyson~\cite{Dys62} considered a Brownian motion on the space of matrices. More precisely, set $b_{i,j}(t):=b_{i,j}^1(t)+\I b_{i,j}^2(t)$, where $b_{i,j}^1(t)$ and $b_{i,j}^2(t)$, $1\leq i,j\leq N$, are independent standard Brownian motions. The matrix $B(t)$ with entries  \mbox{$B_{i,j}(t):=\frac12(b_{i,j}(t)+\overline{b_{j,i}(t)})$} is a matrix-valued Brownian motion on Hermitian matrices. The stationary matrix-valued Ornstein-Uhlenbeck process defined by
\begin{equation}\label{eq2}
\dx H(t)=-\frac{\beta}{4N} H(t) \dx t+\dx B(t),\quad \textrm{with }\beta=2,
\end{equation}
is called \emph{$\beta=2$ Dyson's Brownian motion} and its stationary measure is (\ref{eq1}).

(b) \emph{Symmetric matrices.}
The Gaussian Orthogonal Ensemble is a measure on $N\times N$ symmetric matrices with probability measure as in (\ref{eq1}) but with $\beta=1$ (and, of course, with $\dx H=\prod_{1\leq i\leq j\leq N} \dx H_{i,j}$). Similarly, one defines the $\beta=1$ DBM by (\ref{eq2}) with $\beta=1$ (and $b_{i,j}^2(t)=0$).

\subsubsection*{Continuous time TASEP}
The continuous time TASEP is a Markov process defined on the space \mbox{$\Omega=\{0,1\}^\Z$}. For a configuration $\eta(t)\in \Omega$, we say that at position $j$ and time $t$ there is a particle if $\eta_j(t)=1$, otherwise the position is empty. The dynamics is the following: particles jumps to their neighboring right site with rate $1$, provided the site is empty. Let $f:\Omega\to\R$ be a function depending on a finite number of $\eta_j$'s. Then, the backward generator $L$ of TASEP is given by
\begin{equation}
L f(\eta)=\sum_{j\in\Z} \eta_j(1-\eta_{j+1})\left(f(\eta^{j,j+1})-f(\eta)\right)
\end{equation}
where $\eta^{j,j+1}$ is the configuration $\eta$ with the occupations at sites $j$ and $j+1$ interchanged. $e^{Lt}$ is the transition probability of the TASEP, see~\cite{Li85b,Li99} for more details on the construction. In the following we will discuss results for two specific initial conditions:\\
(a) \emph{step initial conditions}: $\eta_j(0)=1$ for $j<0$ and $\eta_j(0)=0$ for $j\geq 0$,\\
(b) \emph{alternating initial conditions}: $\eta_j(0)=1$ for even $j$ and $\eta_j(0)=0$ for odd~$j$.

\bigskip
In which cases do we have the same limit processes in TASEP and Gaussian random matrices?
The probably most famous result is the convergence to the GUE Tracy-Widom distribution $F_2$: Let $\lambda_{{\rm max,N}}^{\rm GUE}$ the largest eigenvalue of GUE $N\times N$ matrices. Then~\cite{TW94},
\begin{equation}\label{eq4}
\lim_{N\to\infty} \Pb(\lambda_{{\rm max,N}}^{\rm GUE}\leq 2N+s N^{1/3})=F_2(s).
\end{equation}
The analogous result for TASEP occurs for step initial conditions\footnote{A similar result for the partially asymmetric exclusion process has been recently determined~\cite{TW08b}.}. Let $x_n(t)$ denote the position at time $t$ of the particle starting from position \mbox{$x_n(0)=-n$}. Then,
\begin{equation}
\lim_{t\to\infty} \Pb(x_{[t/4]}(t)\leq -s (t/2)^{1/3})=F_2(s).
\end{equation}
This connection extends to joint distributions~\cite{Jo03b,BFS07}. Let $\lambda^{\rm GUE}_{{\rm max},N}(t)$ be the largest eigenvalue of $\beta=2$ Dyson's Brownian motion at time $t$. Then,
\begin{equation}
\begin{aligned}
&\lim_{N\to\infty} \frac{\lambda_{{\rm max,N}}^{\rm GUE}(2uN^{2/3})-2N}{N^{1/3}} = {\cal A}_2(u),\\
&\lim_{t\to\infty} \frac{x_{[t/4+u (t/2)^{2/3}]}(t)+2u(t/2)^{2/3}-u^2(t/2)^{1/3}}{-(t/2)^{1/3}}= {\cal A}_2(u),
\end{aligned}
\end{equation}
in the sense of finite-dimensional distributions, where ${\cal A}_2$ is the Airy$_2$ process (firstly obtained in a stochastic growth model by Pr\"ahofer and Spohn~\cite{PS02}; see also~\cite{Fer07} for a definition and properties).

For GOE matrices, in~\cite{TW96} Tracy and Widom proved that the same rescaling as in (\ref{eq4}) leads to a well-defined limit denoted by $F_1$  and called the GOE Tracy-Widom distribution function:
\begin{equation}
\lim_{N\to\infty} \Pb(\lambda_{{\rm max,N}}^{\rm GOE}\leq 2N+s N^{1/3})=F_1(s).
\end{equation}
An analogous result holds for TASEP with alternating initial conditions. Namely, let $x_n(t)$ be the position at time $t$ of the particle starting from $x_n(0)=-2n$. Then\footnote{This was proven in~\cite{Sas05,BFPS06}, but for a related point-to-line last passage percolation model, it was obtained by Baik and Rains before~\cite{BR99b}, see also~\cite{PS00} for the interpretation as growth process on a flat substrate.},
\begin{equation}
\lim_{t\to\infty} \Pb(x_{[t/4]}(t)\leq -s t^{1/3}/2)=F_1(s).
\end{equation}
One might then hope for an extension to the joint distributions of this relation. As shown in~\cite{Sas05,BFPS06}
\begin{equation}
\lim_{t\to\infty} \frac{x_{[t/4+u t^{2/3}]}(t)+2u(t/2)^{2/3}}{-t^{1/3}}= {\cal A}_1(u),
\end{equation}
in the sense of finite-dimensional distributions, where ${\cal A}_1$ is the Airy$_1$ process (see also~\cite{Fer07} for a definition and properties). However, a convincing numerical evidence~\cite{BFP08} shows that
\begin{equation}\label{eq10}
\lim_{t\to\infty} \frac{\lambda_{{\rm max,N}}^{\rm GOE}(8uN^{2/3})-2N}{2N^{1/3}}\neq {\cal A}_1(u).
\end{equation}
The covariance of ${\cal A}_1(u)$ decays super-exponentially fast in $u$, while the one for the process in l.h.s.\ of (\ref{eq10}) only polynomially.

\newpage
These results give rise to a number of questions:
\begin{itemize}
\item Is this link between TASEP and Gaussian ensembles of random matrices just accidential or do they share a common underlying structure?
\item At which point does this variety of connections come to an end?
\end{itemize}
The link between GOE and TASEP with alternating initial condition seems to be restricted to the static case\footnote{In a related stochastic growth model, the connection extends from the statistics of the largest eigenvalue to the one of the the top eigenvalues~\cite{Fer04}. However, it is restricted to fixed time.}. On the other hand, as we shall discuss below, GUE and TASEP with step initial condition have a much stronger relation, which can be seen comparing a $2+1$ dimensional extension of TASEP with the eigenvalues of the GUE minors. However, also this connection is only partial: The Markov property at the level of eigenvalues' minor does not hold in general as proven in~\cite{ANvM10b} (see also Remark~11.1 in~\cite{Def08}), while it holds for the interacting particle system described below.

\subsubsection*{$2+1$ dynamics on interlaced particle systems}
An extension of TASEP with step initial condition to a dynamics on a set of interlaced particle system has been introduced in~\cite{BF08}. We denote by $x_k^m(t)$ the position at time $t$ of the $k$th leftmost particle at level $m$, $1\leq k \leq m \leq n$. As initial condition we have $x_k^m(0)=k-m-1$ and the configuration space of the system with $n$ levels is
\begin{equation}
{\cal S}^{(n)}=\{x_k^m\in \Z\, |\, x_k^{m+1} < x_k^m\leq x_{k+1}^{m+1}, 1\leq k \leq m\leq n\}.
\end{equation}
The dynamics is as follows: Each particle $x_k^m$ has an independent exponential clock of rate one, and when the $x_k^m$-clock rings, the particle attempts to jump to the right by one. If at that moment $x_k^m=x_k^{m-1}-1$, then the jump is blocked. If that is not the case, we find the largest $c\geq 1$ such that $x_k^m=x_{k+1}^{m+1}=\dots=x_{k+c-1}^{m+c-1}$, and all $c$ particles in this string jump to the right by one.

Both the evolution on ${\cal S}^{(n)}$ and its projection onto $\{x_1^m, m\geq 1\}$ are Markov processes, where the second is nothing else but the TASEP with step initial conditions described above. The space-time correlation functions for this model are not completely known. However, if we restrict ourselves to so-called space-like paths they are determinantal. Introduce the notation
\begin{equation}\label{eqPartialOrder}
(n_1,t_1)\prec (n_2,t_2) \quad\text{iff} \quad n_1\leq n_2,
t_1\geq t_2,\text{ and }(n_1,t_1)\neq (n_2,t_2).
\end{equation}
We say that $(n_1,t_1)$ and $(n_2,t_2)$ are \emph{space-like} if either $(n_2,t_2)\prec (n_1,t_1)$ or $(n_1,t_1)\prec (n_2,t_2)$. Then, a path is called space-like if any two points on it are space-like. The two extreme cases of space-like paths are (1) fixed level $n$ and increasing time $t$ and (2) fixed time $t$ and decreasing level $n$. In~\cite{BF08} it is proven that along space-like paths the correlation functions are determinantal.

\begin{thm}[Theorem~1.1 and Proposition~4.2 of~\cite{BF08}]\label{ThmDetStructure}$ $\\
For any $m=1,2,\dots$, pick $m$ (distinct) triples
\begin{equation}
\varkappa_j=(x_j,n_j,t_j)\in \Z\times\N\times \R_{\geq 0}
\end{equation}
such that
\begin{equation}
t_1\leq t_2\leq \dots\leq t_m,\qquad n_1\geq n_2\geq \dots\geq n_m.
\end{equation}
Then
\begin{multline}
\Pb\{\textrm{For each }j=1,\dots,m \textrm{ there exists a } k_j, \\
 1\leq k_j\leq n_j\textrm{ such that }x^{n_j}_{k_j}(t_j)=x_j\}=\det{[{\cal K}(\varkappa_i,\varkappa_j)]}_{1\leq i,j\leq m},
\end{multline}
where
\begin{multline}
{\cal K}(\varkappa_1;\varkappa_2) =-\frac{1}{2\pi\I}\oint_{\Gamma_{0,1}}\dx w\, \frac{(w-1)^{n_1-n_2}e^{(t_1-t_2)w}}{w^{x_1+n_1-x_2-n_2+1}}\Id_{[(n_1,t_1)\prec (n_2,t_2)]}\\
+\frac{1}{(2\pi\I)^2}\oint_{\Gamma_1}\dx z \oint_{\Gamma_{0,z}}\dx w \, \frac{e^{t_1 w} (1-w)^{n_1}}{w^{x_1+n_1+1}}\frac{z^{x_2+n_2}}{e^{t_2 z}(1-z)^{n_2}}\frac{1}{w-z}
\end{multline}
For a set $A$, $\Gamma_A$ is any simple path positively oriented including as only poles the elements of the set $A$.
\end{thm}
Under the diffusion scaling limit
\begin{equation}
X_k^n(\tau):=\lim_{t\to\infty} \frac{x_k^n(\tfrac12\tau t)-\tfrac12\tau t}{\sqrt{t}}
\end{equation}
one readily obtains that the correlation functions for the $X_k^n$'s are, along space-like paths, still determinantal with kernel
\begin{multline}\label{eq17}
\widetilde{\cal K}(\xi_1,n_1,\tau_1;\xi_2,n_2,\tau_2) =
-\frac{2}{2\pi\I} \int_{\I\R+\e} \dx w\,\frac{e^{(\tau_1-\tau_2)w^2-2(\xi_1-\xi_2)w}}{w^{n_2-n_1}}\Id_{[(n_1,t_1)\prec (n_2,t_2)]}\\
+\frac{2}{(2\pi\I)^2}\oint_{|z|=\e/2}\dx z \int_{\I\R+\e}\dx w\, \frac{e^{\tau_1 w^2-2\xi_1 w}}{e^{\tau_2 z^2-2\xi_2 z}}\frac{w^{n_1}}{z^{n_2}} \frac{1}{w-z}
\end{multline}
where $\e>0$ is arbitrary.
This kernel for $\tau_1=\tau_2=1$ appeared first for the GUE minors~\cite{JN06} and was shown to occur in TASEP in~\cite{BFS07}. An antisymmetric version (for general $\tau$'s) of this kernel was derived in~\cite{BFS09} (with a slightly different scaling in space) and extends the kernel for the antisymmetric GUE minors of~\cite{FN08}.

\subsubsection*{GUE minors}
The normalization in (\ref{eq1}) is the best suited to make the comparison clear between the scaling limits of large matrices and large time in TASEP. However, if we look at matrices of different sizes, it is more natural to drop the $N$-dependence in the Gaussian term of the GUE measure. Thus, for what follows, we consider instead of (\ref{eq1}) the following probability measure on $N\times N$ Hermitian matrices:
\begin{equation}\label{eq1New}
\frac{1}{\widetilde Z_N} \exp\left(-\Tr(H^2)\right) \dx H.
\end{equation}
Denote by $\lambda_k^m$ the $k$th smallest eigenvalue of the principal submatrix obtained from the first $m$ rows and columns of a GUE matrix. In our context, these principal submatrices are usually referred to as minors, and not (as otherwise customary) their determinants. The result is well known, see e.g.~\cite{Bar01,FN08b,Def08}: given the eigenvalues of the $N\times N$ matrix, the GUE minors' eigenvalues are uniformly distributed on the set
\begin{equation}\label{eq18}
{\cal D}^{(N)}=\{\lambda_k^m\in \R\, |\, \lambda_k^{m+1} \leq \lambda_k^m\leq \lambda_{k+1}^{m+1}, 1\leq k \leq m\leq N\}.
\end{equation}
It is proven in~\cite{JN06} that the correlation functions of these eigenvalues are determinantal with correlation kernel
\begin{multline}
K^{\rm GUE}(\xi_1,n_1;\xi_2,n_2)=-\frac{2}{2\pi\I}\int_{\I \R+\e}\dx w\, \frac{e^{-2(\xi_1-\xi_2)w}}{w^{n_2-n_1}}\Id_{[n_1<n_2]}\\
+\frac{2}{(2\pi\I)^2}\oint_{|z|=\e/2} \dx z \int_{\I \R+\e}\dx w \,\frac{e^{w^2-2w \xi_1}}{e^{z^2-2z \xi_2}}\frac{w^{n_1}}{z^{n_2}}\frac{1}{w-z},
\end{multline}
for any $\e>0$.
A way of proving is the following. Obviously, changing the condition $\lambda_k^{m+1} \leq \lambda_k^m$ into $\lambda_k^{m+1} < \lambda_k^m$ does not change the system, since we cut out null sets. Then, using Sasamoto's trick originally employed for TASEP~\cite{Sas05}, one can replace the interlacing condition by a product of determinants,
\begin{equation}
\prod_{m=1}^{N-1} \det[\phi(\lambda_i^{m},\lambda_j^{m+1})]_{1\leq i,j\leq m+1},
\end{equation}
where $\lambda_{m+1}^m\equiv{\rm virt}$ are \emph{virtual variables}, $\phi(x,y)=\Id_{[x\leq y]}$, $\phi({\rm virt},y)=1$. Thus, the measure on ${\cal D}^{(N)}$ becomes
\begin{equation}\label{eq22}
{\rm const} \times \left(\prod_{m=1}^{N-1} \det[\phi(\lambda_i^{m},\lambda_j^{m+1})]_{1\leq i,j\leq m+1}\right) \Delta(\lambda^N)\prod_{i=1}^N e^{-(\lambda_i^N)^2}\,\dx\lambda,
\end{equation}
where $\dx\lambda=\prod_{1\leq k\leq n\leq N}\dx\lambda_k^n$, and $\Delta$ is the Vandermonde determinant. Finally one simply applies Lemma~3.4 of~\cite{BFPS06}. A further approach is presented in~\cite{FN08b}.

\subsection*{Results}
\subsubsection*{Evolution of GUE minors}
Does there exist an extension of such a result for the minors of $\beta=2$ Dyson's Brownian motion? There are two aspects to be considered. The first is to determine whether the evolution of the minors' eigenvalues can be described by a Markov process. It is known that it is not the case if one takes at least three consecutive minors~\cite{ANvM10b}. However, along space-like paths the evolution is indeed Markovian (see Section~\ref{AppMarkov}). The second issue concerns the correlation functions and if they have any similarities with the ones for the $2+1$ particle system defined above. As we shall prove, the answer is affirmative if we restrict ourselves to space-like paths. This is the content of Theorem~\ref{ThmMainGUE} below.

To make the connection more straightforward, we replace the Ornstein-Uhlenbeck processes by Brownian motions starting from $0$. Note that the two models are the same after an appropriate change of scale in space-time. While preparing this manuscript the result analogue to Theorem~\ref{ThmMainGUE} below for the Ornstein-Uhlenbeck case was obtained by Adler, Nordenstam and van~Moerbeke~\cite{ANvM10}.

Let $H(t)$ be an $N\times N$ Hermitian matrix defined by
\begin{equation}
H_{i,j}(t)=
\begin{cases}
\frac{1}{\sqrt{2}}\, b_{i,i}(t), & \textrm{if }1\leq i \leq N, \\
\frac12(b_{i,j}(t)+\I \, \tilde b_{i,j}(t)), & \textrm{if }1\leq i < j \leq N, \\
\frac12(b_{i,j}(t)-\I \, \tilde b_{i,j}(t)), & \textrm{if }1\leq j < i \leq N,
\end{cases}
\end{equation}
where $b_{i,j}(t)$ and $\tilde b_{i,j}(t)$ are independent standard Brownian motions.
The measure on the $N\times N$ matrix at time $t$ is then given by
\begin{equation}
\frac{1}{\widetilde Z_{N,t}} \exp\left(-\frac{\Tr(H^2)}{t}\right) \dx H.
\end{equation}

For $n\in\{1,\dots,N\}$ we denote by $H(n,t)$ the $n\times n$ minor of $H(t)$, which is obtained by keeping the first $n$ rows and columns of $H(t)$. Denote by $\lambda^n_1(t)\leq \lambda_2^n(t)\leq \dots\leq \lambda_n^n(t)$ the eigenvalues of $H(n,t)$. Then, at any time $t$, the interlacing property (\ref{eq18}) holds. Moreover, along space-like paths the eigenvalues' process is Markovian with correlation functions given as follows.
\begin{thm}\label{ThmMainGUE}
For any $m=1,2,\dots$, pick $m$ (distinct) triples
\begin{equation}
\varkappa_j=(x_j,n_j,t_j)\in \R\times\N\times \R_{\geq 0}
\end{equation}
such that
\begin{equation}
t_1\leq t_2\leq \dots\leq t_m,\qquad n_1\geq n_2\geq \dots\geq n_m.
\end{equation}
Then, the $m$-point correlation function of the eigenvalues' point process is given by
\begin{equation}
\rho^{(m)}(\varkappa_1,\dots,\varkappa_m)=\det{[{\cal K}^{\rm GUE}(\varkappa_i,\varkappa_j)]}_{1\leq i,j\leq m},
\end{equation}
where
\begin{multline}\label{eqExtendedGUEkernel}
{\cal K}^{\rm GUE}(\varkappa_1;\varkappa_2) =-\frac{2}{2\pi\I}\int_{\I\R+\e}\dx w\, \frac{e^{(t_1-t_2)w^2-2(x_1-x_2)w}}{w^{n_2-n_1}}\Id_{[(n_1,t_1)\prec (n_2,t_2)]}\\
+ \frac{2}{(2\pi\I)^2}\oint_{|z|=\e/2}\dx z \int_{\I\R+\e}\dx w \, \frac{e^{w^2t_1-2x_1w}}{e^{z^2t_2-2x_2z}}
\frac{1}{w-z}\frac{w^{n_1}}{z^{n_2}}
\end{multline}
where $\e>0$.
\end{thm}

\subsubsection*{Evolution of Wishart minors}
The appearence of determinantal correlation functions along space-like paths is not only limited to Brownian motion on GUE matrices, but they also occur in other Hermitian matrix models, namely the Laguerre Unitary Ensemble. We show that the evolution of Wishart matrices~\cite{TW03b} along space-like paths is determinantal and determine the space-time correlation kernel.

Let $A(n,t)$ be a $p\times n$ complex valued matrix defined by
\begin{equation}
A_{i,j}(n,t)=\frac{1}{\sqrt{2}}(b_{i,j}(t)+\I\, \tilde b_{i,j}(t)),\quad 1\leq i \leq p, 1\leq j \leq n
\end{equation}
where the $b_{i,j}$'s and $\tilde b_{i,j}$'s are independent standard Brownian motions. Then, we define the (complex) $n\times n$ sample covariance matrix (or Wishart matrix) by $H(n,t)=A(n,t)^*A(n,t)$, which is usually referred to as the Laguerre process. As before, denote by $\lambda_k^n(t)$ the $k$th smallest eigenvalue of $H(n,t)$.
\begin{thm}\label{ThmMainWishart}
For any $m=1,2,\dots$, pick $m$ (distinct) triples
\begin{equation}
\varkappa_j=(x_j,n_j,t_j)\in \R\times\{1,\dots,p\}\times \R_{\geq 0}
\end{equation}
such that
\begin{equation}
t_1\leq t_2\leq \dots\leq t_m,\qquad n_1\geq n_2\geq \dots\geq n_m.
\end{equation}
Then, the $m$-point correlation function of the eigenvalues' point process is given by
\begin{equation}
\rho^{(m)}(\varkappa_1,\dots,\varkappa_m)=\det{[{\cal K}^{\rm LUE}(\varkappa_i,\varkappa_j)]}_{1\leq i,j\leq m},
\end{equation}
where
\begin{multline}
{\cal K}^{\rm LUE}(\varkappa_1;\varkappa_2) =-\frac{1}{2\pi\I} \oint_{\Gamma_0}\dx z\, \frac{e^{x_1/(z-t_1)}}{e^{x_2/(z-t_2)}}\frac{(z-t_1)^{p-1-n_1}}{(z-t_2)^{p+1-n_2}}\Id_{[(n_1,t_1)\prec (n_2,t_2)]}\\
+ \frac{-1}{(2\pi\I)^2}\oint_{\Gamma_0}\dx z \oint_{\Gamma_{z,t_2}}\dx w \, \frac{e^{x_2/(z-t_1)}}{e^{x_2/(w-t_2)}}\frac{(z-t_1)^{p-1-n_1}}{(w-t_2)^{p+1-n_2}}\frac{w^p}{z^p}\frac{1}{w-z}.
\end{multline}
\end{thm}

\subsubsection*{Acknowledgemnts}
This work was supported by the DFG (German Research Foundation) through the SFB (Collaborative Research Center) 611, project A12.

\section{Proof of Theorem~\ref{ThmMainGUE}}
The issue of the Markov property is discussed in Section~\ref{AppMarkov} below and therefore we assume it to hold in this section. For $0<t_1<t_2$, the joint distribution of $H_1=H(n,t_1)$ and $H_2=H(n,t_2)$ is given by
\begin{equation}\label{eq32}
{\rm const}\times \exp\left(-\frac{\Tr(H_1^2)}{t_1}\right) \exp\left(-\frac{\Tr((H_2-H_1)^2)}{t_2-t_1}\right)\dx H_1\,\dx H_2.
\end{equation}
The measure on eigenvalues is obtained using Eynard-Mehta formula~\cite{EM97} for coupled random matrices, which on its turn is based on the Harish-Chandra/Itzykson-Zuber formula~\cite{HC57,IZ80} (see Appendix~\ref{AppHCIZformulas}). It results in the following formula.
\begin{lem}\label{LemGUEFixedn}
Let $n$ be fixed. Denote by $\lambda_{k}^n(t)$, $1\leq k \leq n$, the eigenvalues of $H(n,t)$. Their joint distribution at $0<t_1<t_2$ is given by
\begin{multline}
{\rm const}\times \Delta(\lambda^n(t_1)) \det\left(e^{-(\lambda_{i}^n(t_1)-\lambda_{j}^n(t_2))^2/(t_2-t_1)}\right)_{1\leq i,j\leq n} \Delta(\lambda^n(t_2)) \\
\times\prod_{i=1}^Ne^{-(\lambda_{i}^n(t_1))^2 /t_1} \, \dx \lambda_{i}^n(t_1)\,\dx \lambda_{i}^n(t_2),
\end{multline}
with $\Delta$ the Vandermonde determinant and $\lambda^n(t)=(\lambda^n_1(t),\dots,\lambda_n^n(t))$.
\end{lem}

The second formula concerns the joint distribution of the eigenvalues at two different levels. This result is a special case of the formula (\ref{eq22}) discussed above. (It is enough to reintegrate out the lower levels, which gives a Vandermonde determinant).
\begin{lem}\label{LemGUEFixedt}
Let $t$ be fixed. Denote by $\lambda_{k}^n(t)$, $1\leq k \leq n$, the eigenvalues of $H(n,t)$. Their joint distribution at levels $n$ and $n+1$ is given by
\begin{multline}
{\rm const}\times \Delta(\lambda^{n}(t))
\det[\phi(\lambda_i^{n}(t),\lambda_j^{n+1}(t))]_{1\leq i,j\leq n+1}
\Delta(\lambda^{n+1}(t)) \\
\times\prod_{i=1}^{n+1}e^{-(\lambda_{i}^{n+1}(t))^2/t} \, \dx \lambda_{i}^n(t)\,\dx \lambda_{i}^{n+1}(t),
\end{multline}
where $\lambda_{n+1}^n\equiv{\rm virt}$ are \emph{virtual variables}, $\phi(x,y)=\Id_{[x\leq y]}$, $\phi({\rm virt},y)=1$ (and $\Delta$ the Vandermonde determinant).
\end{lem}

The eigenvalues' process is a Markov process (see Section~\ref{AppMarkov} for details) for both fixed matrix dimension $n$ and increasing time $t$, as well as for fixed time $t$ and decreasing matrix dimension $n$. The combination of the formulas in Lemma~\ref{LemGUEFixedn} and Lemma~\ref{LemGUEFixedt} leads to Proposition~\ref{PropMeasureGUE}:
\begin{prop}\label{PropMeasureGUE}
Let $N_1\geq \dots \geq N_m= 1$ be integers and $0< t_1< \dots < t_m$ be reals. We denote by $\lambda_1^n(t)<\dots<\lambda_n^n(t)$ the eigenvalues of $H(n,t)$ and set $N_0=N_1$, $N_{m+1}=0$. Then the joint density of
\begin{equation}
\{\lambda_k^n(t_j) : 1\leq j\leq m, N_j\leq n \leq N_{j-1}, 1\leq k \leq n\}
\end{equation}
is given by
\begin{multline}\label{eqGUE}
{\rm const} \times  \det\bigl[ \Psi_{N_1-\ell}^{N_1,t_1}(\lambda_k^{N_1}(t_1))\bigr]_{1\leq k,\ell \leq N_1} \\
\times \prod_{j=1}^{m-1} \biggl[ \det \bigl[ {\cal T}_{t_{j+1},t_j}(\lambda_k^{N_j}(t_{j+1}),\lambda_\ell^{N_j}(t_j))\bigr]_{1\leq k,\ell \leq N_j} \\
\times \prod_{n=N_{j+1}+1}^{N_j} \det \bigl[ \phi(\lambda^{n-1}_k(t_{j+1}),\lambda^n_\ell(t_{j+1}))\bigr]_{1\leq k,\ell \leq n}\biggr],
\end{multline}
where
\begin{equation}\label{eq33}
\begin{aligned}
\phi(x,y)&=\Id_{[x\leq y]}, \quad \phi(x_n^{n-1},y)=1, \\
{\cal T}_{t,s}(x,y)&= \frac{1}{\sqrt{\pi(t-s)}} \exp\left(-\frac{(x-y)^2}{t-s}\right)\Id_{[t\geq s]}, \\
\Psi_{k}^{N_1,t_1}(x)&= \frac{1}{t_1^{k/2}}\,p_k\left(\frac{x}{\sqrt{t_1}}\right)\frac{1}{\sqrt{\pi t_1}}\,\exp\left(-\frac{x^2}{t_1}\right),
\end{aligned}
\end{equation}
for $k=0,\dots,N_1-1$. Here $p_k$ is the standard Hermite polynomial of degree $k$ (see Appendix~\ref{AppHermite} for details).
\end{prop}
We could have chosen any polynomials of degree $k$ multiplied by the Gaussian weight without changing the probability measure (\ref{eqGUE}) since the modifications would just affect the normalization constant. However, this choice allows a huge simplification of the computations, because of the properties of Lemma~\ref{lem:1} below.

To determine the kernel, we first slightly rewrite (\ref{eqGUE}). For $1\leq n \leq N_1$ let $c(n)=\#\{i: N_i=n\}$, and we denote the consecutive times for such a level by $t_1^n < \dots < t^n_{c(n)}$. Then, the measure (\ref{eqGUE}) can be rewritten as
\begin{multline}\label{eq36}
{\rm const} \times \prod_{n=2}^{N_1} \Biggl( \det\bigr[\phi(\lambda_k^{n-1}(t_1^{n-1}),\lambda_\ell^n(t^n_{c(n)}))\bigr]_{1\leq k,\ell\leq n} \\
 \times \prod_{a=2}^{c(n)} \det\bigl[ {\cal T}_{t_a^n,t_{a-1}^n} (\lambda_k^n(t_a^n),\lambda_\ell^n(t^n_{a-1}))\bigr]_{1\leq k,\ell\leq n}\Biggr) \det[\Psi_{N_1-\ell}^{N_1,t_1^{N_1}}(\lambda_k^{N_1}(t_1^{N_1}))\bigr]_{1\leq k,\ell\leq N_1}.
\end{multline}
It is known that a measure of this form has determinantal correlations and the correlation kernel is computed by means of Theorem~4.2 of~\cite{BF07}, which we report in Appendix~\ref{AppCorrelation} for the reader.

For any given $k\in \Z$ we set
\begin{equation}\label{eqPsiDefinGUE}
\Psi^{n,t}_k(x)=\frac{2^{k+1}}{t^{(k+1)/2}}\frac{1}{2\pi\I} \int_{\I\R+\e} \dx w \, e^{w^2-2wx/\sqrt{t}} w^k,\quad \e>0.
\end{equation}
For $n=N_1$, $t=t_1$ and $k=0,\dots,N_1-1$, this function is the one in the measure (\ref{eqGUE}), which is obtained from the first representation of Hermite polynomials in (\ref{eqReprHermite}).
\begin{lem} \label{lem:1}
It holds, for $0<r<s<t$ and $k\geq 1$,
\begin{itemize}
\item[(i)] $\phi \ast \Psi^{n,t}_{n-k} = \Psi^{n-1,t}_{n-1-k}$,
\item[(ii)] ${\cal T}_{t,s} \ast \Psi^{n,s}_{n-k} =\Psi^{n,t}_{n-k}$,
\item[(iii)] $\phi \ast {\cal T}_{t,s} = {\cal T}_{t,s} \ast \phi$,
\item[(iv)] ${\cal T}_{t,s} \ast {\cal T}_{s,u}={\cal T}_{t,u}$.
\end{itemize}
\end{lem}
\begin{proof}
For the first relation, we use $\Re(w)=\e>0$ so that we can exchange the two integrals,
\begin{equation}
\begin{aligned}
(\phi \ast \Psi^{n,t}_k)(x)&=\int_x^\infty\dx y\, \frac{2^{k+1}}{t^{(k+1)/2}}\frac{1}{2\pi\I} \int_{\I\R+\e} \dx w \,e^{w^2-2wy/t^{1/2}} w^k\\
&=\frac{2^{k+1}}{t^{(k+1)/2}}\frac{1}{2\pi\I} \int_{\I\R+\e} \dx w \,e^{w^2} w^k \int_x^\infty\dx y \,e^{-2wy/t^{1/2}}\\
&=\frac{2^k}{t^{k/2}}\frac{1}{2\pi\I} \int_{\I\R+\e} \dx w \,e^{w^2-2wx/t^{1/2}} w^{k-1} =\Psi^{n-1,t}_{k-1}(x).
\end{aligned}
\end{equation}

For the second identity, we first do the change of variable $w=z (s/t)^{1/2}$ in the integral representation (\ref{eqPsiDefinGUE}) of $\Psi_k^{n,s}$ and then perform a Gaussian integration:
\begin{equation}
\begin{aligned}
({\cal T}_{t,s} \ast \Psi^{n,s}_k)(x) &= \frac{2^{k+1}}{t^{(k+1)/2}}\frac{1}{2\pi\I}\int_{\I\R+\e}\dx z \, e^{z^2 s/t}z^k \int_\R\dx y\, \frac{\exp\left(-\frac{(x-y)^2}{t-s}-\frac{2yz}{\sqrt{t}}\right)}{\sqrt{\pi(t-s)}}\\
&=\frac{2^{k+1}}{t^{(k+1)/2}}\frac{1}{2\pi\I}\int_{\I\R+\e}\dx z \, e^{z^2} e^{-2xz/\sqrt{t}} z^k =\Psi^{n,t}_k(x).
\end{aligned}
\end{equation}

The third relation is also easy to verify. Indeed,
\begin{equation}
\begin{aligned}
&(\phi \ast {\cal T}_{t,s})(x,z)=\int_\R \dx y \, \phi(x,y){\cal T}_{t,s}(y,z)=\int_{\R_+}\dx y \, {\cal T}_{t,s}(y+x,z)\\
&=\int_\R \dx y \, {\cal T}_{t,s}(x,z-y) \phi(z-y,z)=\int_\R \dx y \, {\cal T}_{t,s}(x,y) \phi(y,z) = ({\cal T}_{t,s} \ast \phi)(x,z).
\end{aligned}
\end{equation}

The last relation is the standard heat kernel semigroup identity.
\end{proof}

By Theorem~\ref{ThmPushASEP} and Remark~\ref{RemarkKernel}, there is a simple way of getting the kernel if the matrix $M$ with
\begin{equation}
  M_{k,\ell}  = (\phi \ast {\cal T}^k \ast \cdots \ast \phi \ast {\cal T}^{N_1} \ast \Psi_{N_1-\ell}^{N_1,t_1^{N_1}})(x_k^{k-1}),
\end{equation}
is upper triangular, where ${\cal T}^n := {\cal T}_{t_{c(n)}^n,t_1^n}$. The identities in Lemma~\ref{lem:1} give, for $k\geq \ell$,
\begin{equation}
 M_{k,\ell} = (\phi \ast \Psi^{k,t_{c(k)}^k}_{k-\ell})(x_k^{k-1}) = \int_{\R} \dx x \, \Psi^{k,t_{c(k)}^k}_{k-\ell}(x)
\begin{cases}
= 0,& \textrm{ for }\ell<k,\\
\neq 0, &\textrm{ for }\ell=k,
\end{cases}
\end{equation}
because the last expression is (after a rescaling in $x$) proportional to the orthogonal relation (\ref{eq57}) for $n=0$ and $m=k-\ell$.

Next we need to determine the polynomials $\Phi^{n,t}_\ell(x)$, $\ell=0,\dots,n-1$, which are biorthogonal to the functions $\Psi^{n,t}_k(x)$, $k=0,\dots,n-1$, i.e., polynomials satisfying
\begin{equation}\label{eq41}
\int_{\R}\dx x\, \Psi^{n,t}_k(x)\Phi^{n,t}_\ell(x) =\delta_{k,\ell},\quad 1\leq k,\ell\leq n-1.
\end{equation}

\begin{lem}
The functions
\begin{equation}\label{eqPhiDefinGUE}
\Phi^{n,t}_\ell(x)=\frac{1}{\ell!} \frac{t^{\ell/2}}{2^\ell}p_\ell\Big(\frac{x}{\sqrt{t}}\Big) =\frac{t^{\ell/2}}{2^\ell}\frac{1}{2\pi\I}\oint_{\Gamma_0}\dx z \,\frac{e^{-z^2+2zx/t^{1/2}}}{z^{\ell+1}}
\end{equation}
satisfy the relation (\ref{eq41}).
\end{lem}
\begin{proof}
One does the change of variable $x\mapsto x\sqrt{t}$ and then uses the orthogonal relation (\ref{eq57}).
\end{proof}

Let us compute the last term in (\ref{eq60}). To simplify the notations, we set $t_1=t_{a_1}^{n_1}$ and $t_2=t_{a_2}^{n_2}$. First, we do the changes of variables $w=\sqrt{t_1} \tilde w$ and $z=\sqrt{t_2}\tilde z$ in (\ref{eqPsiDefinGUE}) and (\ref{eqPhiDefinGUE}). We obtain
\begin{equation}
\begin{aligned}
\sum_{k=1}^{n_2} \Psi^{n_1,t_1}_{n_1-k}(x_1) \Phi^{n_2,t_2}_{n_2-k}(x_2) & =
\sum_{k=1}^{n_2} \frac{2^{n_1}}{2^{n_2}}\frac{2}{(2\pi\I)^2}\int_{\I\R+\e}\dx \tilde w \oint_{\Gamma_0}\dx \tilde z \frac{e^{\tilde w^2 t_1-2\tilde w x_1}}{e^{\tilde z^2 t_2-2\tilde z x_2}}\frac{\tilde w^{n_1-k}}{\tilde z^{n_2+1-k}}
\end{aligned}
\end{equation}
Now, we take the integral over $\tilde z$ to satisfy $|\tilde z|<|\tilde w|$, say $|\tilde z|=\e/2$. This allows us to take the sum inside and extend it to $+\infty$ (because for $k>n_2$ the pole at zero for $\tilde z$ vanishes). The sum over $k$ gives
\begin{equation}
\sum_{k\geq 1} \frac{\tilde z^{k-1}}{\tilde w^k} = \frac{1}{\tilde w-\tilde z}
\end{equation}
so that we obtain
\begin{multline} \label{eq54}
\sum_{k=1}^{n_2} \Psi^{n_1,t_1}_{n_1-k}(x_1) \Phi^{n_2,t_2}_{n_2-k}(x_2) \\
= \frac{2^{n_1}}{2^{n_2}}\frac{2}{(2\pi\I)^2}\int_{\I\R+\e}\dx \tilde w \oint_{|z|=\e/2}\dx \tilde z \, \frac{e^{\tilde w^2 t_1-2\tilde w x_1}}{e^{\tilde z^2 t_2-2\tilde z x_2}}\frac{\tilde w^{n_1}}{\tilde z^{n_2}}\frac{1}{\tilde w-\tilde z}.
\end{multline}

The last term we have to compute is $\phi^{(t_{a_1}^{n_1},t_{a_2}^{n_2})}$, see Theorem~\ref{ThmDetStructure}. To simplify the notations, we set $\phi^{(t_{a_1}^{n_1},t_{a_2}^{n_2})}(x,y)=\phi^{(n_1,t_1;n_2,t_2)}(x,y)$. We have
\begin{equation}
\phi^{(n_1,t_1;n_2,t_2)} =
\begin{cases}
\phi^{\ast (n_2-n_1)} \ast {\cal T}_{t_2,t_1}, & \text{ if } (n_1,t_1)\prec (n_2,t_2), \\ 0, & \text{ otherwise}.
\end{cases}
\end{equation}
It is easy to verify that $\phi(x,y)$ has the integral representation
\begin{equation}
\phi(x,y)=\frac{2}{2\pi \I}\int_{\I\R+\e}\dx w \, \frac{e^{2w(y-x)}}{2w},\quad \e>0.
\end{equation}
and similarly,
\begin{equation}
\phi^{\ast n}(x,y)=\frac{2}{2\pi \I}\int_{\I\R+\e}\dx w\, \frac{e^{2w(y-x)}}{(2w)^n},\quad \e>0.
\end{equation}
Then, for $(n_1,t_1)\prec (n_2,t_2)$, a Gaussian integration gives us
\begin{equation}\label{eq51}
\phi^{(n_1,t_1;n_2,t_2)}(x_1,x_2)=\frac{2^{n_1}}{2^{n_2}}\frac{2}{2\pi\I}\int_{\I\R+\e}\dx w\, \frac{e^{w^2(t_1-t_2)-2w(x_1-x_2)}}{w^{n_2-n_1}}.
\end{equation}

Equations (\ref{eq54}) and (\ref{eq51}) yield a kernel which is, up to the conjugation factor\footnote{A determinantal point process is defined by its correlation kernel, which is defined up to conjugations.} $2^{n_1-n_2}$, the same as (\ref{eqExtendedGUEkernel}). Thus the proof of Theorem~\ref{ThmMainGUE} is completed.

\section{Proof of Theorem~\ref{ThmMainWishart}}
As for the GUE case, the issue of the Markov property is discussed in Section~\ref{AppMarkov} below. For $0<t_1<t_2$, the joint distribution of $A_1=A(n,t_1)$ and $A_2=A(n,t_2)$ is given by
\begin{equation}\label{eq32b}
{\rm const}\times \exp\left(-\frac{\Tr(A_1^*A_1)}{t_1}\right) \exp\left(-\frac{\Tr((A_2^*-A_1^*)(A_2-A_1))}{t_2-t_1}\right)\dx A_1\,\dx A_2.
\end{equation}
The measure on eigenvalues is obtained (as in the Ornstein-Uhlenbeck case studied in~\cite{TW03b}) by the Harish-Chandra/Itzykson-Zuber formula for rectangular matrices~\cite{JSV96,ZJZ03} (see Appendix~\ref{AppHCIZformulas}). It results in the following formula.

\begin{lem}\label{LemLUEFixedn}
 Let $n$ be fixed. Denote by $\lambda_k^n(t)$, $1\leq k \leq n\leq p$, the eigenvalues of $H(n,t)=A(n,t)^*A(n,t)$. Their joint distribution at $0<t_1<t_2$ is given by
\begin{multline}
{\rm const} \\
\times \det\left[I_{p-n}\left(\tfrac{2 \sqrt{\lambda^n_i(t_1)\lambda^n_j(t_2)}}{t_2-t_1}\right) \left(\tfrac{\lambda_j^n(t_2)}{\lambda_i^n(t_1)}\right)^{(p-n)/2} e^{-(\lambda_i^n(t_1)+\lambda_j^n(t_2))/(t_2-t_1)}\right]_{1\leq i,j\leq n} \\
\times \Delta(\lambda^n(t_1)) \Delta(\lambda^n(t_2)) \prod_{i=1}^n (\lambda_i(t_1))^{p-n} e^{-\lambda_i^n(t_1)/t_1}\, \dx\lambda_i^n(t_1)\,\dx\lambda_i^n(t_2),
\end{multline}
where $I_m$ is the modified Bessel function of order $m$, see (\ref{eq91}).
\end{lem}

The second formula concerns the joint distributions of the eigenvalues at two different levels. This is studied in~\cite{FN08b} with the following result.
\begin{lem}\label{LemLUEFixedt}
 Let $t$ be fixed. Denote by $\lambda_k^n(t)$, $1\leq k\leq n<p$, the eigenvalues of $H(n,t)$. Their joint distribution at levels $n$ and $n+1$ is given by
\begin{multline}
  {\rm const} \times \Delta(\lambda^n(t)) \det\bigl[\phi(\lambda_i^n(t),\lambda_j^n(t))\bigr]_{1\leq i,j\leq n+1} \Delta(\lambda^{n+1}(t)) \\
\times \prod_{i=1}^{n+1} (\lambda_i^{n+1}(t))^{p-(n+1)} e^{-\lambda_i^{n+1}(t)/t}\,\dx \lambda_i^n(t) \, \dx \lambda_i^{n+1}(t),
\end{multline}
where $\lambda_{n+1}^n \equiv \virt$ are virtual variables, $\phi(x,y)=\Id_{[x\geq y]}$ and $\phi(\virt,y)=1$.
\end{lem}

Putting together the formulas in lemmata~\ref{LemLUEFixedn} and~\ref{LemLUEFixedt} leads to the next proposition.

\begin{prop}\label{PropMeasureLUE}
Let $p\geq N_1\geq \dots \geq N_m= 1$ be integers and \mbox{$0< t_1< \dots< t_m$} be real numbers. We denote by $\lambda_1^n(t)<\dots<\lambda_n^n(t)$ the eigenvalues of $H(n,t)$ and set $N_0=N_1$, $N_{m+1}=0$. Then the joint density of
\begin{equation}
 \left\{ \lambda_k^n(t_j) : 1\leq j\leq m, N_j\leq n \leq N_{j-1}, 1\leq k \leq n\right\}
\end{equation}
is given by
\begin{multline}\label{eq58}
 {\rm const} \times  \det\bigl[ \Psi_{N_1-\ell}^{p-N_1,t_1}\bigl(\lambda_k^{N_1}(t_1)\bigr)\bigr]_{1\leq k,\ell \leq N_1} \\
 \times \prod_{j=1}^{m-1} \biggl[ \det \bigl[ {\cal T}_{t_{j+1},t_j}^{p-N_j}\bigl(\lambda_k^{N_j}(t_{j+1}),\lambda_\ell^{N_j}(t_j)\bigr)\bigr]_{1\leq k,\ell \leq n_j} \\
 \times \prod_{\ell=N_{j+1}+1}^{N_j} \det \bigl[ \phi\bigl(\lambda^{n-1}_k(t_{j+1}),\lambda^{n}_\ell(t_{j+1})\bigr)\bigr]_{1\leq k,\ell \leq n}\biggr],
\end{multline}
where
\begin{equation}\label{eq61}
\begin{aligned}
\phi(x,y)&=\Id_{[x\geq y]}\quad \textrm{and}\quad \phi(\lambda_{n+1}^n,y)=1,\\
{\cal T}_{t,s}^n (x,y) &= \left(\frac{x}{y}\right)^{n/2} I_n\left( \frac{2\sqrt{xy}}{t-s}\right)  \frac{1}{t-s} \exp\left(-\frac{x+y}{t-s}\right) \Id_{[x,y>0]}\Id_{[s\leq t]}, \\
\Psi_k^{p-N_1,t_1} (x) & = \frac{k!}{(p-N_1+k)!t_1^{k+1}} \, \left(\frac{x}{t_1}\right)^{p-N_1} \exp\left(-\frac{x}{t_1}\right) L_k^{p-N_1} \left(\frac{x}{t_1}\right) \Id_{[x>0]},
\end{aligned}
\end{equation}
for $k=0,\dots, N_1-1$. Here $L^n_k$ are the generalized Laguerre polynomials of order $n$ and degree $k$, see Appendix~\ref{AppLaguerre}.
\end{prop}
Comparing the mathematical structure of (\ref{eqGUE}) and (\ref{eq58}), we see that the only difference is that the transition kernel for time depends also on the level. However, this does not pose any problem, see Remark~\ref{RemarkTimeDepT}.

For $k \in \mathbb Z$ and $x\in\R_+$ we set
\begin{equation}\label{eq62}
\Psi_k^{n,t}(x)=\frac{t^{-(k+1)}}{2\pi\I} \oint_{\Gamma_0}\dx z\,\frac{(z-1)^k}{z^{n+k+1}} e^{x(z-1)/t}.
\end{equation}
For $n=p-N_1$, $t=t_1$ and $k=0,\dots,N_1-1$ the above defined function coincides with (\ref{eq61}). Moreover, the prefactors are chosen such that the following nice recursion relations hold.

\newpage
\begin{lem} \label{lem:2}
It holds, for $t>s>r>0$, $n \leq p$, and $k\geq 1$
\begin{itemize}
\item[(i)] $\phi \ast \Psi_{n-k}^{p-n,t} = \Psi_{(n-1)-k}^{p-(n-1),t}$,
\item[(ii)] ${\cal T}^{p-n}_{t,s} \ast \Psi_{n-k}^{p-n,s} = \Psi_{n-k}^{p-n,t}$,
\item[(iii)] $\mathcal \phi \ast {\cal T}_{t,s}^{p-n} = {\cal T}_{t,s}^{p-(n-1)}\ast \phi$,
\item[(iv)] ${\cal T}^{p-n}_{t,s} \ast {\cal T}^{p-n}_{s,r}={\cal T}^{p-n}_{t,r}$.
\end{itemize}
\end{lem}
To prove this lemma, we first obtain a different integral representation for (\ref{eq62}). Namely, after the change of variable $z=\tilde z/(\tilde z-t)$ we get
\begin{equation}\label{eqPsiLag}
\Psi_k^{n,t}(x)=\frac{-1}{2\pi\I} \oint_{\Gamma_0}\dx \tilde z \, \frac{(\tilde z-t)^{n-1}}{\tilde z^{n+k+1}} e^{x/(\tilde z-t)}.
\end{equation}

\begin{proof}[Proof of Lemma~\ref{lem:2}]
Using the representation (\ref{eqPsiLag}), we have
\begin{equation}
\begin{aligned}
(\phi \ast \Psi_k^{n,t})(x) &= \frac{-1}{2\pi \I} \oint_{\Gamma_0} \dx z\,\frac{(z-t)^{n-1}}{z^{n+k+1}} \int_0^x \dx y \, e^{y/(z-t)} \\
& = \frac{-1}{2\pi \I} \oint_{\Gamma_0} \dx z \, \frac{(z-t)^n}{z^{n+k+1}} \left(e^{x/(z-t)}-1\right) = \Psi_{k-1}^{n+1,t}(x)
\end{aligned}
\end{equation}
because for $k\geq 0$ the term independent of $x$ has residue equal to zero.

Using the integral representation (\ref{eq91}) of the modified Bessel function $I_n$ in (\ref{eq58}), we get (for $x,y>0, t>s>0$)
\begin{equation}\label{eqTLagB}
{\cal T}_{t,s}^n(x,y) = \frac{1}{2\pi \I (t-s)} \oint_{\Gamma_0} \frac{\dx z}{z^{n+1}} \exp\left(-\frac{x(1-z)+y(1-z^{-1})}{t-s}\right),
\end{equation}
and the change of variable $z=(w-s)/(w-t)$ leads to
\begin{equation}\label{eqTLag}
{\cal T}_{t,s}^n(x,y) = \frac{-1}{2\pi \I} \oint_{\Gamma_s} \dx w \, \frac{(w-t)^{n-1}}{(w-s)^{n+1}} e^{x/(w-t)-y/(w-s)}.
\end{equation}

We can choose the integration path with $w$ large and $z$ small so that $\Re(1/(z-s)-1/(w-t))<0$ (in particular, $z$ is contained in $\Gamma_s$, so that we write it explicitly as $\Gamma_{s,z}$). Then, we can exchange the integral over $y$ with the integral over $z$ and $w$,
\begin{equation}
\begin{aligned}
& ({\cal T}^n_{t,s} \ast \Psi_k^{n,s})(x) \\
& = \frac{1}{(2\pi \I)^2} \oint_{\Gamma_0} \dx z \frac{(z-s)^{n-1}}{z^{n+k+1}} \oint_{\Gamma_{s,z}} \dx w \frac{(w-t)^{n-1}}{(w-s)^{n+1}} e^{x/(w-t)} \int_{\R_+} \dx y\, e^{y/(z-s)-y/(w-s)} \\
& = \frac{1}{(2\pi \I)^2} \oint_{\Gamma_0} \dx z \, \frac{(z-s)^n}{z^{n+k+1}} \oint_{\Gamma_{s,z}} \dx w \, \frac{(w-t)^{n-1}}{(w-s)^n} e^{x/(w-t)} \, \frac{1}{z-w}.
\end{aligned}
\end{equation}
Now we enlarge the path of $z$ so that encloses the path of $w$. This can be made at the expense of the residue at $z=w$. Thus we get
\begin{equation}
\begin{aligned}
({\cal T}^n_{t,s} \ast \Psi_k^{n,s})(x) &= \frac{1}{(2\pi \I)^2} \oint_{\Gamma_{s}} \dx w \, \frac{(w-t)^{n-1}}{(w-s)^n} e^{x/(w-t)} \oint_{\Gamma_{0,w}} \dx z \, \frac{(z-s)^n}{z^{n+k+1}}\, \frac{1}{z-w}\\
&\qquad -\frac{1}{2\pi\I}\oint_{\Gamma_0}\dx w \,\frac{(w-t)^{n-1}}{w^{n+k+1}}e^{x/(w-t)}=\Psi_k^{n,t}(x),
\end{aligned}
\end{equation}
because the first term is zero, since the residue of $z$ at infinity is zero ($k\geq 0$).

For the third identity, we use the representation (\ref{eqTLagB}) in which we take the path $\Gamma_0$ for $w$ to satisfy $|w|>1$. Then,
\begin{equation}
\begin{aligned}
(\phi \ast {\cal T}_{t,s}^{n-1})(x,y)&= \frac{1}{2\pi \I(t-s)} \oint_{\Gamma_0} \frac{\dx w}{w^n} \,e^{-\frac{y}{t-s}(1-w^{-1})} \int_0^x\dx z\, e^{-\frac{z}{t-s} (1-w)} \\
& = \frac{1}{2 \pi \I} \oint_{\Gamma_0} \frac{\dx w}{(w-1)w^n}\,e^{-\frac{y}{t-s}(1-w^{-1})}\left( e^{-\frac{x}{t-s} (1-w)} -1\right).
\end{aligned}
\end{equation}
The last term (the integrand independent of $x$) is zero, because the integrand has residue zero at infinity, whenever $n-1\geq 0$.
Thus,
\begin{equation}
\begin{aligned}
(\phi \ast {\cal T}_{t,s}^{n-1})(x,y)&= \frac{1}{2 \pi \I} \oint_{\Gamma_0} \frac{\dx w}{(w-1)w^n}\,e^{-\frac{x}{t-s} (1-w)-\frac{y}{t-s}(1-w^{-1})} \\
& = \frac{1}{2\pi \I (t-s)} \oint_{\Gamma_0} \frac{\dx w}{w^{n+1}} \, e^{-\frac{x}{t-s}(1-w)} \int_y^\infty \dx z \, e^{-\frac{z}{t-s}(1-w^{-1})}\\
& = ({\cal T}_{t,s}^n \ast \phi)(x,y),
\end{aligned}
\end{equation}

The final identity is true because ${\cal T}^{n}$ it is the transition density of a $2n+1$ dimensional Bessel process.
\end{proof}

We proceed as in the proof of Theorem~\ref{ThmMainGUE} to show that the matrix $M$ is upper triangular. Indeed, with $\mathcal T^n := \mathcal T^{p-n}_{t^n_{c(n)},t^n_1}$ and Lemma~\ref{lem:2} we find
\begin{equation}
M_{k,\ell} = (\phi \ast \Psi_{k-\ell}^{p-k,t_{c(k)}^k})(x_k^{k-1})= \int_{\R_+} \dx x \, \Psi_{k-\ell}^{p-k,t_{c(k)}^k}(x) =
\begin{cases} 0, & \text{if } \ell < k, \\ 1, & \text{if } \ell = k,
\end{cases}
\end{equation}
because of the orthogonality between $\Psi^{n,t}_k$, $k\geq 1$, and the constant function.

\begin{lem} Define, for $\ell=0,\dots,n-1$, the polynomial $\Phi_\ell^{n,t}$ of degree $\ell$ by
\begin{equation}\label{eq68}
\Phi_\ell^{n,t}(x) = \frac{1}{2\pi\I} \oint_{\Gamma_{0,t}} \dx w\, \frac{w^{n+\ell}}{(w-t)^{n+1}}\,e^{-x/(w-t)}.
\end{equation}
These polynomials satisfy the orthogonal relation
\begin{equation}\label{eq68b}
\int_{\R_+}\dx x\, \Psi^{n,t}_k(x)\Psi^{n,t}_\ell(x)=\delta_{k,\ell}
\end{equation}
for $k,\ell=0,\dots,n-1$.
\end{lem}
\begin{proof}
By the integral representation in Appendix~\ref{AppLaguerre} for Laguerre polynomials, we have
\begin{equation}
t^\ell L_{\ell}^n(x/t) = \frac{t^\ell}{2\pi\I} \oint_{\Gamma_1} \dx w\, \frac{w^{n+\ell}}{(w-1)^{\ell+1}}\,e^{-x(w-1)/t}
\end{equation}
and, after change of variable $w=\tilde w/(\tilde w-t)$, we get $t^\ell L_{\ell}^n(x/t) = \Phi_\ell^{n,t}(x)$ as defined in (\ref{eq68}).
The orthogonality relation (\ref{eq68b}) holds because after the change of variable $x\to xt$, the left-hand side becomes
\begin{equation}
t \int_{\R_+}\dx x\, \Psi^{n,t}_k(x t)\Phi^{n,t}_\ell(x t)=\frac{k! t^\ell}{(n+k)!t^{k}} \int_{\R_+}\dx x\, x^n e^{-x} L_\ell^n(x) L_k^n(x) =\delta_{k,\ell},
\end{equation}
which is the orthogonal relation (\ref{OrthoLaguerre}) for Laguerre polynomials.
\end{proof}

We now compute the kernel and start with the sum in (\ref{eq60}). Let us use the notations  $t_1=t_{a_1}^{n_1}$, $t_2=t_{a_2}^{n_2}$. Then we get
\begin{multline}
\sum_{k=1}^{n_2} \Psi_{n_1-k}^{p-n_1,t_1}(x_1) \Phi_{n_2-k}^{p-n_2,t_2}(x_2) \\
= \frac{-1}{(2\pi\I)^2} \oint_{\Gamma_0} \dx z \oint_{\Gamma_{0,t_2}}\dx w \, \frac{e^{x_1/(z-t_1)}}{e^{x_2/(w-t_2)}} \, \frac{(z-t_1)^{p-n_1-1}}{(w-t_2)^{p-n_2+1}}\,\frac{w^p}{z^{p+1}} \sum_{k=1}^{n_2} \left( \frac{z}{w}\right)^k.
\end{multline}
We choose $\Gamma_0$ and $\Gamma_{0,t_2}$ such that they do not intersect, i.e., $|z|<|w|$. For $k>n_2$ the pole at $w=\infty$ vanishes and we can thus extend the summation over $k$ to $\infty$ with the result
\begin{equation}
\frac{-1}{(2\pi\I)^2} \oint_{\Gamma_0} \dx z \oint_{\Gamma_{z,t_2}}\dx w \, \frac{e^{x_1/(z-t_1)}}{e^{x_2/(w-t_2)}} \, \frac{(z-t_1)^{p-n_1-1}}{(w-t_2)^{p-n_2+1}}\,\frac{w^p}{z^p} \, \frac{1}{w-z},
\end{equation}
which is the second term in the kernel in Theorem~\ref{ThmMainWishart}.

It remains to compute $\phi^{(t_{a_1}^{n_1},t_{a_2}^{n_2})}$. To simplify the notations, we set $\phi^{(t_{a_1}^{n_1},t_{a_2}^{n_2})}(x,y)=\phi^{(n_1,t_1;n_2,t_2)}(x,y)$. By Lemma~\ref{lem:2} we have
\begin{equation}
\phi^{(n_1,t_1;n_2,t_2)} =
\begin{cases}
{\cal T}_{t_2,t_1}^{p-n_1} \ast \phi^{\ast(n_2-n_1)}, & \text{if } (n_1,t_1) \prec (n_2,t_2), \\ 0, & \text{otherwise}.
\end{cases}
\end{equation}
The integral representation (\ref{eqTLagB}) for ${\cal T}$ and $\phi^{\ast n}(x,y)=\frac{(x-y)^{n-1}}{(n-1)!}\phi(x,y)$ lead to
\begin{multline}
\phi^{(n_1,t_1;n_2,t_2)}(x,y)\\
= \frac{(t_1-t_2)^{-1}}{2\pi \I} \oint_{\Gamma_0} \dx w\,\frac{e^{-(1-w)x/(t_1-t_2)}}{w^{p+1-n_1}} \int_y^\infty \dx z\,e^{-z(1-w^{-1})/(t_1-t_2)} \frac{(z-y)^{n_2-n_1-1}}{(n_2-n_1-1)!} \\
= \frac{(t_1-t_2)^{n_2-n_1-1}}{2\pi \I} \oint_{\Gamma_0} \dx w\, \frac{e^{-x(1-w)/(t_1-t_2)-y(1-w^{-1})/(t_1-t_2)}}{w^{p+1-n_2}(w-1)^{n_2-n_1}}.
\end{multline}
Finally, the change of variable $w=(z-t_2)/(z-t_1)$ gives the first term in Theorem~\ref{ThmMainWishart}.

\section{Markov property on space-like paths}\label{AppMarkov}
The process on matrices is clearly a Markov process along space-like paths. What we have to see is that the Markov property still holds for the eigenvalues. The key ingredients are that the measure on matrices is invariant under choice of basis, and that the choice of basis at an observation point $(n,t)$ depends neither on the eigenvalues at that the previous point ($(n+1,t)$ or $(n,t')$ with $t'<t$) nor on the eigenvalues at $(n,t)$.

\subsubsection*{Diffusion on GUE minors}
We first consider Dyson's diffusion. Here we denote by $H(n,t)$ the \mbox{$n\times n$} minor at time $t$ and by $\Lambda(n,t)$ the diagonalized matrix of $H(n,t)$ which is obtained from conjugation by the unitary matrix $U(n,t)$,
\begin{equation}\label{eq76}
H(n,t)=U(n,t)\Lambda(n,t)U^*(n,t).
\end{equation}
The Jacobian of the transformation $H(n,t)\to(\Lambda(n,t),U(n,t))$ gives
\begin{equation}\label{eq78}
\dx H(n,t)=\Delta(\Lambda(n,t))^2\,\dx\Lambda(n,t)\,\dx\mu_n(U(n,t))
\end{equation}
where $\dx\mu_n$ denotes the Haar measure on the unitary group ${\cal U}(n)$.

Consider a measure at $(n,t)$ which is invariant under unitary transformations, i.e.\ with respect to the group ${\cal U}(n)$. It has the form
\begin{equation}\label{MP0}
f_1(\Lambda(n,t))\, \dx\Lambda(n,t)\, \dx\mu_n(U(n,t))
\end{equation}
for some explicit function $f_1$ (e.g.\ $f_1(H)=\exp(-\Tr(H^2)/t)$).

Next fix $n$ and consider times $t'>t$. Then, the probability measure on the matrices has the form (see (\ref{eq32}))
\begin{equation}
\begin{aligned}
& f_1(\Lambda(n,t)) e^{-b\Tr((H(n,t')-H(n,t)))^2}\,\dx\Lambda(n,t)\,\dx\mu_n(U(n,t)) \,\dx H(n,t')\\
= {} & f_1(\Lambda(n,t)) e^{-b\Tr((\Lambda(n,t')-\tilde U(n,t)\Lambda(n,t)\tilde U^*(n,t)))^2}\,\dx\Lambda(n,t)\,\dx\mu_n(\tilde U(n,t)) \,\dx H(n,t')
\end{aligned}
\end{equation}
because of the unitary invariance of the Haar measure (here we have set $\tilde U(n,t)=U(n,t')^* U(n,t)$). The integration with respect to $\dx\mu_n(\tilde U(n,t))$ is made by the well-known Harish-Chandra/Itzykson-Zuber (\ref{eqIZHC1}) and the result is as in Lemma~\ref{LemGUEFixedn}. We are left with a probability density that depends only on the eigenvalues times $\dx H(n,t')$, that is, \emph{the projection onto eigenvalues at time $t$ did not restrict the complete freedom of choice of basis at $(n,t')$}. Otherwise stated, by the decomposition (\ref{eq78}), after integration over $\dx\mu_n(\tilde U(n,t))$ we have a measure on eigenvalues times $\dx\mu_n(U(n,t'))$ of the form (for some explicit $f_2$, which can be easily computed)
\begin{equation}\label{MP1}
f_2(\Lambda(n,t),\Lambda(n,t'))\,\dx\Lambda(n,t)\,\dx\Lambda(n,t')\,\dx\mu_n(U(n,t')).
\end{equation}

The other choice is to consider $t$ fixed and look at the measure at $(n,t)$ and $(n-1,t)$. The result explained in Proposition~4.2 of~\cite{Bar01} is actually a conditional measure on eigenvalues given the eigenvalues of the minor of size $n$, thus it is not restricted to GUE, but it holds for any measure which is invariant under ${\cal U}(n)$, see Theorem~3.4 of~\cite{Def08}. The projection on the eigenvalues at $(n,t)$ and $(n-1,t)$ leads to Lemma~\ref{LemGUEFixedt}. We can also decide to project on the eigenvalues at $(n,t)$ and $(n-1,t)$ and the eigenvectors at $(n-1,t)$. This means that we do not integrate out the variables corresponding to the unitary transformations of the $(n-1)\times (n-1)$ minor given by
\begin{equation}
\left(
  \begin{array}{cc}
    U(n-1,t) & 0 \\
    0 & 1 \\
  \end{array}
\right),\quad \textrm{with }U(n-1,t)\textrm{ distributed as }\dx\mu_{n-1},
\end{equation}
which form a subgroup of ${\cal U}(n)$. The eigenvalues $\Lambda(n-1,t)$ are independent of the eigenvectors (thus of the choice of basis $U(n-1,t)$) and the measure on $U(n-1,t)$ is $\dx\mu_{n-1}(U(n-1,t))$ (see e.g.~Corollary~2.5.4 in~\cite{AGZ10}). The measure on $H(n,t)$ is invariant under ${\cal U}(n)$, so are the eigenvalues $\Lambda(n,t)$ independent of the choice of $U(n-1,t)$ (this last property follows also from the direct computation in Section~3.1 of~\cite{FN08b}; Section~3.2 for Wishart matrices). Thus, the projection on the eigenvalues at $(n,t)$ and $(n-1,t)$ and the eigenvectors at $(n-1,t)$ leads to a measure of the form
\begin{equation}\label{MP2}
f_3(\Lambda(n,t),\Lambda(n-1,t))\,\dx \Lambda(n,t) \,\dx\Lambda(n-1,t) \,\dx\mu_{n-1}(U(n-1,t)).
\end{equation}
for some explicit function $f_3$ (compare with Lemma~\ref{LemGUEFixedt}).

To resume, (\ref{MP1}) and (\ref{MP2}) tell us that starting from a measure of the form (\ref{MP0}), in which the choice of basis is completely free, the projection onto the eigenvalues obtained by integration over the angular variables does not fix the basis at the next step in the basic steps of space-like paths. This implies that the eigenvalues' process along space-like paths is a Markov process.

\subsubsection*{Diffusion on Wishart minors}
Consider now diffusion on Wishart matrices. Let $H(n,t)=A^*(n,t)A(n,t)$ be the \mbox{$n\times n$} minor at time $t$, where $A(n,t)$ is the $p\times n$ matrix with singular value decomposition $A(n,t)=U(p,t)\Sigma(n,t)V^*(n,t)$, where $U(p,t)$ is a $p\times p$ Haar-distributed on ${\cal U}(p)$, $V(n,t)$ is a $n\times n$ Haar-distributed on ${\cal U}(n)$, and $\Sigma(n,t)$ is a $p\times n$ matrix with entries zeros except on the diagonal, where we find the singular values of $A(n,t)$. Also, let $\Lambda(n,t)=\Sigma^*(n,t)\Sigma(n,t)$ the matrix of the eigenvalues of $H(n,t)$. Thus we have
\begin{equation}\label{eq76b}
H(n,t)=A^*(n,t)A(n,t)=V(n,t)\Lambda(n,t)V^*(n,t).
\end{equation}
The Jacobian of the transformation $A(n,t)\to(\Sigma(n,t),U(n,t),U(p,t))$ gives (see e.g.~\cite{Mui82})
\begin{multline}
 \dx A(n,t)={\rm const}\times (\det(\Sigma^\ast(n,t)\Sigma(n,t)))^{p-n+1/2}\Delta^2(\Sigma^\ast(n,t)\Sigma(n,t)) \\
\times \dx \Sigma(n,t)\,\dx \mu_n(V(n,t))\, \dx\mu_p(U(p,t)),
\end{multline}
or, using that $\Lambda(n,t) = \Sigma^\ast(n,t)\Sigma(n,t)$,
\begin{multline}\label{eq78b}
\dx A(n,t)={\rm const}\times (\det(\Lambda(n,t)))^{p-n}\Delta^2(\Lambda(n,t))\\
\times\dx\Lambda(n,t)\,\dx\mu_n(V(n,t))\,\dx\mu_p(U(p,t)).
\end{multline}
Therefore, the starting measure at $(n,t)$ has the form
\begin{equation}\label{MP0b}
g_1(\Lambda(n,t))\, \dx\Lambda(n,t)\, \dx\mu_n(V(n,t)) \, \dx\mu_p(U(p,t))
\end{equation}
for some explicit function $g_1$.

Next consider fixed $n$ and time $t'>t$. Then, the probability measure on the matrices has the form (see (\ref{eq32b}))
\begin{equation}
\begin{aligned}
& g_1(\Lambda(n,t)) \,e^{-b\Tr((A^*(n,t')-A^*(n,t))(A(n,t')-A(n,t)))}\\
& \qquad \qquad \times\dx\Lambda(n,t)\,\dx\mu_n(V(n,t)) \,\dx\mu_p(U(p,t))\,\dx A(n,t') \\
={}&g_1(\Lambda(n,t)) \, e^{-b\Tr\left([\Sigma^*(n,t')-\tilde V^*(n,t)\Sigma^*(n,t)\tilde U(p,t)]\, [\Sigma(n,t')-\tilde U(p,t)\Sigma(n,t)\tilde V^*(n,t)]\right)}\\
&\qquad \qquad \times \,\dx\Lambda(n,t)\,\dx\mu_n(\tilde V(n,t)) \,\dx\mu_p(\tilde U(p,t))\,\dx A(n,t')
\end{aligned}
\end{equation}
because of unitary invariance of the Haar measure (we set $\tilde V(n,t)=V(n,t')^*V(n,t)$ and $\tilde U(p,t)=U(p,t')^*U(p,t)$). An integration with respect to $\dx\mu_n(\tilde V(n,t))\,\dx\mu_p(\tilde U(p,t))$ according to (\ref{eqIZHC2}) results in the formula of Lemma~\ref{LemLUEFixedn}. We are left with a probability density that depends only on the eigenvalues times $\dx A(n,t')$, that is, the projection onto eigenvalues at time $t$ did not restrict the complete freedom of choice of basis at $(n,t')$. Otherwise stated, by (\ref{eq78b}) we have a measure on eigenvalues times $\dx\mu_n(V(n,t'))\,\dx\mu_p(U(p,t'))$ of the form (for some explicit $g_2$, which can be easily computed)
\begin{equation}\label{MP1b}
g_2(\Lambda(n,t),\Lambda(n,t'))\,\dx\Lambda(n,t)\,\dx\Lambda(n,t')\,\dx\mu_n(V(n,t'))\,\dx\mu_p(U(p,t')).
\end{equation}

The other choice is to consider $t$ fixed and look at the measure at $(n,t)$ and $(n-1,t)$. This works as for the Hermitian case and we get a measure of the form
\begin{equation}\label{MP2b}
g_3(\Lambda(n,t),\Lambda(n-1,t))\,\dx \Lambda(n,t) \,\dx\Lambda(n-1,t) \,\dx\mu_{n-1}(V(n-1,t))\,\dx\mu_p(U(p,t)).
\end{equation}
for some explicit function $g_3$ (compare with Lemma~\ref{LemLUEFixedt}).

Therefore (\ref{MP1b}) and (\ref{MP2b}) tell us that starting from a measure of the form (\ref{MP0b}), in which the choice of basis (in which the matrix $A$ is represented) is completely free, the projection onto the eigenvalues obtained by integration over the angular variables does not fix the basis at the next step in the basic steps of space-like paths. This implies that the eigenvalues' process along space-like paths is a Markov process.

\appendix

\section{Space-like determinantal correlations}\label{AppCorrelation}
For convenience we report here the statement of Theorem~4.2 of~\cite{BF07}.

Let $\X_1,\dots,\X_N$ be finite sets and $c(1),\dots,c(N)$ be arbitrary nonnegative integers. Consider the set
\begin{equation}
\X=(\X_1\sqcup\dots\sqcup\X_1)\sqcup\dots\sqcup(\X_N\sqcup \dots\sqcup\X_N)
\end{equation}
with $c(n)+1$ copies of each $\X_n$. We want to consider a special form of the weight $W(X)$ for any subset $X\subset\X$, which turns out to have determinantal correlations.

To define the weight we need a bit of notations. Let
\begin{equation}
\begin{aligned}
\phi_n(\,\cdot\,,\,\cdot\,):\X_{n-1}\times\X_{n}\to \C,\qquad &n=2,\dots,N, \\
\phi_n(\virt,\,\cdot\,):\X_{n}\to\C,\qquad &n=1,\dots,N,\\
\Psi^N_j(\,\cdot\,):\X_N\to\C,\qquad &j=0,\dots,N-1,
\end{aligned}
\end{equation}
be arbitrary functions on the corresponding sets. Here the symbol
$\virt$ stands for a ``virtual'' variable, which is convenient to
introduce for notational purposes. In applications $\virt$ can
sometimes be replaced by $+\infty$ or $-\infty$. The $\phi_n$ represents the transitions from $\X_{n-1}$ to $\X_n$.

Also, let
\begin{equation}
t_{0}^N\leq\dots\leq t_{c(N)}^N= t_{0}^{N-1}\leq\dots\leq
t_{c(N-1)}^{N-1}=t_0^{N-2}\leq \dots\leq t^2_{c(2)}=t^1_0\leq\dots \leq
t^1_{c(1)}
\end{equation}
be real numbers. In applications, these numbers refer to time moments.
Finally, let
\begin{equation}
{\cal T}_{t_a^n,t_{a-1}^n}(\,\cdot\,,\,\cdot\,):\X_n\times\X_n\to \C,\qquad n=1,\dots,N,\quad a=1,\dots,c(n),
\end{equation}
be arbitrary functions. The ${\cal T}_{t_a^n,t_{a-1}^n}$ represents the transition between two copies of $\X_n$ associated to ``times'' $t^n_{a-1}$ and $t^n_a$.

Then, to any subset $X\subset\X$ we assign its weight $W(X)$ as follows. $W(X)$ is zero unless $X$ has exactly $n$ points in each copy of $\X_n$, $n=1,\dots,N$. In the latter case, denote the points of $X$ in the $m$th copy of
$\X_n$ by $x^n_k(t^n_m)$, $k=1,\dots,n$, $m=0,\dots,c(n)$. Thus,
\begin{equation}
X=\{x^n_k(t^n_m)\mid k=1,\dots,n;\,m=0,\dots,c(n);\, n=1,\dots,N\}.
\end{equation}
Set
\begin{multline}\label{eqAppDetMeasure}
W(X)=  \prod_{n=1}^{N} \Bigg[\det{\bigl[\phi_n(x_k^{n-1}(t_0^{n-1}),x_l^n(t^n_{c(n)}))\bigr]}_{1\leq k,l\leq n} \\
\times\prod_{a=1}^{c(n)} \det{\bigl[{\cal T}_{t_a^n,t_{a-1}^n}(x_k^n(t^n_a),x^n_l(t^n_{a-1}))\bigr]}_{1\leq k,l\leq n}
\Bigg]\det{\bigl[\Psi^{N}_{N-l}(x^{N}_k(t_0^{N}))\bigr]}_{1\leq k,l\leq N},
\end{multline}
where $x^{n-1}_{n}(\,\cdot\,)=\virt$ for all $n=1,\dots,N$.

In what follows we assume that the partition function of our weights
does not vanish,
\begin{equation}
Z:=\sum_{X\subset \X} W(X)\ne 0.
\end{equation}
Under this assumption, the normalized weights $\widetilde W(X)=W(X)/Z$
define a (generally speaking, complex valued) measure on $2^\X$
of total mass $1$. One can say
that we have a (complex valued) random point process on $\X$, and
its correlation functions are defined accordingly, see e.g.,~\cite{RB04}.
We are interested in computing these correlation functions.

Let us introduce the compact notation for the convolution of several transitions. For any $n=1,\dots,N$ and two time moments $t_a^n>t_b^n$ we define
\begin{equation}
{\cal T}_{t_a^n,t_b^n}={\cal T}_{t_a^n,t_{a-1}^n}*{\cal T}_{t_{a-1}^n,t_{a-2}^n}*\cdots *{\cal T}_{t_{b+1}^n,t_b^n},\qquad {{\cal T}}^n = {{\cal T}}_{t_{c(n)}^n,t_0^n},
\end{equation}
where we use the notation $(f*g)(x,y):=\sum_{z}f(x,z)g(z,y).$
For any time moments
$t_{a_1}^{n_1}\geq t_{a_2}^{n_2}$ with $(a_1,n_1)\ne (a_2,n_2)$, we denote the convolution over all the transitions between them by $\phi^{(t_{a_1}^{n_1},t_{a_2}^{n_2})}$:
\begin{equation}
\phi^{(t_{a_1}^{n_1},t_{a_2}^{n_2})}={{\cal T}}_{t^{n_1}_{a_1},t^{n_1}_{0}} * \, \phi_{n_1+1}*{{\cal T}}^{n_1+1} *\cdots*\phi_{n_2}*{{\cal T}}_{t^{n_2}_{c(n_2)},t^{n_2}_{a_2}}.
\end{equation}
If there are no such transitions, i.e., if $t_{a_1}^{n_1}<t_{a_2}^{n_2}$ or $(a_1,n_1)=(a_2,n_2)$, we set
$\phi^{(t_{a_1}^{n_1},t_{a_2}^{n_2})}=0$.

Furthermore, define the matrix $M={\Vert M_{k,l}\Vert}_{k,l=1}^N$ by
\begin{equation}
M_{k,l}=\big(\phi_{k}*{{\cal T}}^k*\cdots *\phi_{N}*{{\cal T}}^{N}*\Psi^{N}_{N-l}\big)(\virt)
\end{equation}
and the vector
\begin{equation}
\Psi^{n,t^n_a}_{n-l}=\phi^{(t^n_a,t^{N}_0)}*\Psi^{N}_{N-l},\qquad l=1,\dots,N.
\end{equation}
The following statement describing the correlation kernel is a part of Theorem 4.2 of~\cite{BF07}.

\begin{thm}[Part of Theorem 4.2 of~\cite{BF07}]\label{ThmPushASEP}
Assume that the matrix $M$ is invertible. Then \mbox{$Z=\det M\neq 0$}, and the (complex valued) random point process on $\X$ defined by the weights $\widetilde W(X)$ is determinantal. Its correlation kernel can be written in the form
\begin{multline}
K(n_1,t^{n_1}_{a_1},x_1; n_2,t^{n_2}_{a_2},x_2)= -\phi^{(t^{n_1}_{a_1},t^{n_2}_{a_2})}(x_1,x_2) \\
+ \sum_{k=1}^{N} \sum_{l=1}^{n_2} \Psi^{n_1,t^{n_1}_{a_1}}_{n_1-k}(x_1) [M^{-1}]_{k,l} (\phi_l * \phi^{(t^l_{c(l)},t^{n_2}_{a_2})})(\virt,x_2).
\end{multline}
\end{thm}

\begin{remark}\label{RemarkKernel}
As stated in the complete statement of Theorem 4.2 of~\cite{BF07}, there is one situation where the kernel takes a simple formula. Namely, when the matrix $M$ is upper-triangular, then
\begin{equation}
\Phi^{n_2,t_{a_2}^{n_2}}_{n_2-k}(x):=\sum_{l=1}^{n_2} [M^{-1}]_{k,l} (\phi_l * \phi^{(t^l_{c(l)},t^{n_2}_{a_2})})(\virt,x)
\end{equation}
are the function biorthogonal to $\Psi^{n_2,t_{a_2}^{n_2}}_{n_2-k}(x)$ obtained for the non-extended kernel (i.e., at fixed level and fixed time). In the case of random matrices which we consider, the functions $\Phi^{n,t}_k$, $k=0,\dots,n-1$, have to be polynomials of degree $k$ because $\det(\Phi^{n,t}_k(x_j))_{1\leq j,k\leq n}$ must be proportional to \mbox{$\Delta(x)=\det(x_j^{k-1})_{1\leq j,k\leq n}$}, the Vandermonde determinant. Then, the kernel is simply written as
\begin{equation}\label{eq60}
K(n_1,t^{n_1}_{a_1},x_1; n_2,t^{n_2}_{a_2},x_2)= -\phi^{(t^{n_1}_{a_1},t^{n_2}_{a_2})}(x_1,x_2)
+ \sum_{k=1}^{n_2} \Psi^{n_1,t^{n_1}_{a_1}}_{n_1-k}(x_1) \Phi^{n_2,t_{a_2}^{n_2}}_{n_2-k}(x_2).
\end{equation}
\end{remark}

\begin{remark}\label{RemarkTimeDepT}
Looking at the proof of the above theorem in~\cite{BF07} one also sees that the time evolutions $\cal T$ can be taken to be level-inhomogeneous, i.e., the ${\cal T}_{t_a^n,a_0^n}$ can be a function of $t_a^n, t_0^n$ and also of the level $n$. Such a situation occurs for the Wishart matrices case.
\end{remark}

The proof of Theorem~\ref{ThmPushASEP} given in~\cite{BF07} is based on the
algebraic formalism of~\cite{RB04}. Another proof can be found in Section 4.4 of~\cite{FN08}.
Although we stated Theorem~\ref{ThmPushASEP} for the case when all sets $\X_n$
are finite, one easily extends it to a more general setting. Indeed,
the determinantal formula for the correlation functions is an
algebraic identity, and the limit transition to the case when the $\X_n$'s
are allowed to be countably infinite is immediate, under the
assumption that all the sums needed to define the $*$-operations
above are absolutely convergent.

\section{Hermite polynomials}\label{AppHermite}
The Hermite polynomial of degree $n$ is denoted here $p_n(x)$. We use the normalization of~\cite{KS96}:
\begin{equation}\label{eq57}
\int_\R\dx x\, e^{-x^2} p_n(x) p_m(x)=\delta_{m,n} \sqrt{\pi} 2^n n!.
\end{equation}
There are two useful integral representations for the Hermite polynomials $p_n(x)$,
\begin{equation}\label{eqReprHermite}
\begin{aligned}
p_n(x)&=\frac{2^n}{\I\sqrt{\pi}}e^{x^2}\int_{\I\R+\e}\dx w\, e^{w^2-2xw}w^n,\\
p_n(x)&=\frac{n!}{2\pi\I}\oint_{\Gamma_0}\dx z\, e^{-(z^2-2xz)} z^{-(n+1)},
\end{aligned}
\end{equation}
as well as the identities (with $0<q<1$) which can be found in~\cite{Jo04,KS96}
\begin{equation}\label{eq8.7}
\begin{aligned}
\frac{1}{\sqrt{\pi(1-q^2)}}\exp\left(-\frac{(x-qy)^2}{1-q^2}\right) &= e^{-x^2}\sum_{k=0}^\infty \frac{p_k(x)p_k(y) q^k}{\sqrt{\pi} 2^{k} k!},\\
\int_x^\infty \dx y\, e^{-y^2} p_n(y)&= e^{-x^2} p_{n-1}(x),\\
p_n(x)&=(-1)^n p_n(-x).
\end{aligned}
\end{equation}
These identities can be useful to rewrite the double integral representation into an expression in terms of Hermite polynomials (as it was made e.g.\ in Lemma~24 of~\cite{BFS09} for the antisymmetric GUE minor kernel).

\section{Laguerre polynomials}\label{AppLaguerre}
The generalized Laguerre polynomials $L^p_k$ of degree $k$ and order $p$ are polynomials on $\R_+$ defined by
\begin{equation}
L^p_k(x)=\frac{x^{-p} e^x}{k!} \frac{\dx^k}{\dx x^k} (x^{p+k}e^{-x}).
\end{equation}
They satisfy the orthogonal relation
\begin{equation}\label{OrthoLaguerre}
\int_{\R_+} \dx x \, x^p e^{-x} L^p_k(x)L^p_\ell(x) = \frac{(p+k)!}{k!}\,\delta_{k,\ell}
\end{equation}
and have integral representations,
\begin{equation}
\begin{aligned}
L_k^p(x) &= \frac{1}{2\pi\I} \oint_{\Gamma_1} \dx w\, \frac{e^{-x (w-1)}w^{p+k}}{(w-1)^{k+1}}, \\
L_k^p(x) &= \frac{(p+k)!}{k!x^p} \frac{1}{2\pi \I} \oint_{\Gamma_0}\dx z\, \frac{e^{xz}(z-1)^k}{z^{p+k+1}}.
\end{aligned}
\end{equation}

\section{Harish-Chandra/Itzykson-Zuber formulas}\label{AppHCIZformulas}
Here we report the Harish-Chandra/Itzykson-Zuber formula as well as its generalization for rectangular matrices.

Let $A=\diag(a_1,\dots,a_N)$ and $B=\diag(b_1,\dots,b_N)$ two diagonal $N\times N$ matrices. Let $\dx\mu$ denote the Haar measure on the unitary group ${\cal U}(N)$. Then,
\begin{equation}\label{eqIZHC1}
\int_{{\cal U}(N)}\dx \mu(U) \exp\left(\Tr(A U B U^*)\right) =
\prod_{p=1}^{N-1} p! \,\frac{\det\left(e^{a_i b_j}\right)_{1\leq i,j\leq N}}{\Delta(a)\Delta(b)},
\end{equation}
where $\Delta(a)$ is the Vandermonde determinant of the vector $a=(a_1,\dots,a_N)$.

The extension to rectangular matrices can be found in section~3.2 of~\cite{ZJZ03} and was derived in~\cite{JSV96}. Let $A$ be a complex $N_1\times N_2$ matrix, $B$ a complex $N_2\times N_1$ matrix so that the $N_2\times N_2$ matrices $A^* A$ and $B B^*$ are diagonal with (real positive) eigenvalues $a=(a_1,\dots,a_{N_2})$ and $b=(b_1,\dots,b_{N_2})$ respectively. W.l.o.g.\ we assume $N_1\geq N_2$. Then,
\begin{multline}\label{eqIZHC2}
\int_{{\cal U}(N_2)}\dx\mu (U) \int_{{\cal U}(N_1)} \dx\mu(V) \exp\left(\Tr(A U B V^* + B^* U^* A^* V)\right)  \\
=\frac{\prod_{p=1}^{N_2-1} p!\prod_{q=1}^{N_1-1} q!}{\prod_{r=1}^{N_1-N_2-1} r!} \frac{\det\left(I_{N_1-N_2}(2\sqrt{a_i b_j})\right)_{1\leq i,j\leq N_2}}{\Delta(a)\Delta(b)\prod_{i=1}^{N_2} (a_i b_i)^{(N_1-N_2)/2}},
\end{multline}
where $I_n$ is the modified Bessel function defined by
\begin{equation}\label{eq91}
I_n(2x) =\frac{1}{2\pi \I}\oint_{\Gamma_0}\dx z\, \frac{e^{x(z+z^{-1})}}{z^{n+1}}
=\sum^{\infty}_{k=0}\frac{x^k}{k!}\frac{x^{k+|n|}}{(k+|n|)!},
\end{equation}
for $n\in\Z$.


\end{document}